\documentclass[conference]{IEEEtran}

\usepackage{algorithm}
\usepackage{algpseudocode }
\usepackage[geometry]{ifsym}
\usepackage{ifsym}
\usepackage{amssymb}
\usepackage{amscd}
\usepackage{dsfont}
\usepackage{amsmath}
\usepackage{epsfig}
\usepackage{graphics}
\usepackage{psfrag}
\usepackage{rotating}
\usepackage{amsmath}
\usepackage{amsfonts}
\usepackage{url}
\usepackage{color}
\usepackage[caption=false,font=footnotesize]{subfig}
\usepackage{epstopdf}
\usepackage{amsthm}
\usepackage{pgfplots}
\usepackage{tikz}

\usepackage{mathtools}

\usetikzlibrary{arrows,shapes,decorations,backgrounds}

\newtheorem{theorem}{Theorem}
\newtheorem{lemma}[theorem]{Lemma}
\newtheorem{example}{Example}
\newtheorem{definition}{Definition}

\newtheorem{corollary}[theorem]{Corollary}

\newcommand{\cond}{\,\vert\,}

\newcommand{\defeq}{\triangleq}

\newfont{\bbb}{msbm10 scaled 500}

\newfont{\bb}{msbm10 scaled 1100}
\newcommand{\CC}{\mbox{\bb C}}
\newcommand{\RR}{\mbox{\bb R}}

\newcommand{\FF}{\mbox{\bb F}}

\newcommand{\EE}{\mbox{\bb E}}

\newcommand{\av}{{\bf a}}

\newcommand{\dv}{{\bf d}}
\newcommand{\ev}{{\bf e}}

\newcommand{\hv}{{\bf h}}

\newcommand{\rv}{{\bf r}}

\newcommand{\xv}{{\bf x}}
\newcommand{\yv}{{\bf y}}

\newcommand{\onev}{{\bf 1}}

\newcommand{\Id}{{\bf I}}

\newcommand{\Qm}{{\bf Q}}
\newcommand{\Rm}{{\bf R}}
\newcommand{\Sm}{{\bf S}}

\newcommand{\Zm}{{\bf Z}}

\newcommand{\Cc}{{\cal C}}

\newcommand{\Hc}{{\cal H}}
\newcommand{\Ic}{{\cal I}}
\newcommand{\Jc}{{\cal J}}
\newcommand{\Kc}{{\cal K}}

\newcommand{\Nc}{{\cal N}}

\newcommand{\Qc}{{\cal Q}}

\newcommand{\alphav}{\hbox{\boldmath$\alpha$}}

\newcommand{\nuv}{\hbox{\boldmath$\nu$}}
\newcommand{\muv}{\hbox{\boldmath$\mu$}}

\newcommand{\sigmav}{\hbox{\boldmath$\sigma$}}

\renewcommand{\arg}{{\hbox{arg}}}

\DeclareFontFamily{U}{cmfi}{}
\DeclareFontShape{U}{cmfi}{m}{n}{ <-> cmfi10 }{}
\DeclareSymbolFont{CMFI}{U}{cmfi}{m}{n}

\def\argmax{\mathop{\rm argmax}}

\renewcommand{\Qm}{\pmb{Q}}
\renewcommand{\Rm}{\pmb{R}}
\renewcommand{\Sm}{\pmb{S}}

\renewcommand{\Zm}{\pmb{Z}}

\renewcommand{\av}{\pmb{a}}

\renewcommand{\dv}{\pmb{d}}
\renewcommand{\ev}{\pmb{e}}

\renewcommand{\hv}{\pmb{h}}

\renewcommand{\rv}{\pmb{r}}

\renewcommand{\xv}{\pmb{x}}
\renewcommand{\yv}{\pmb{y}}

\newcommand{\gp}[1]{{#1}}

\newcommand{\bigO}[1]{\ensuremath{\mathop{}\mathopen{}O\mathopen{}\left(#1\right)}}

\IEEEoverridecommandlockouts

\begin{document}
\title{Alpha Fair Coded Caching}
\author{\IEEEauthorblockN{Apostolos Destounis$^{1}$,
		Mari Kobayashi$^2$, 
		Georgios Paschos $^1$,
		Asma Ghorbel$^2$ 
		\\}
		\IEEEauthorblockA{$^1$ France Research Center, Huawei Technologies Co. Ltd., email: firstname.lastname@huawei.com
		}
	\IEEEauthorblockA{$^2$Centrale-Sup\'elec, 
		France, email: firstname.lastname@centralesupelec.fr
	}
}
\maketitle

\begin{abstract}
The performance of existing \emph{coded caching} schemes is  sensitive to worst channel quality, a problem which is exacerbated when communicating over fading channels.
In this paper we address this limitation in the following manner:  \emph{in short-term}, we allow transmissions to subsets of users with good channel quality, avoiding users with fades, while \emph{in long-term} we ensure fairness across the different users. 
Our online scheme combines (i) joint scheduling and power control for the broadcast channel with fading, and (ii) congestion control for ensuring the optimal long-term average performance.
We restrict the caching operations to the decentralized scheme of \cite{maddah2013decentralized}, 
and subject to this restriction we prove that our scheme  has near-optimal overall performance with respect to the convex alpha-fairness coded caching optimization.
By tuning the coefficient  alpha, the operator can differentiate user performance with respect to  video delivery rates achievable  by coded caching.
 We demonstrate via simulations our scheme's  superiority  over legacy coded caching and  unicast opportunistic scheduling, which are identified as special cases of our general framework.
\end{abstract}

\section{Introduction}\label{sec:intro}
A key challenge \gp{for the}  future wireless networks is the increasing video traffic demand, which reached 70\% of total mobile IP traffic in 2015 \cite{cisco15}.
Classical downlink systems cannot meet this demand since they have limited resource blocks, and therefore as the number of simultaneous video transfers  $K$ increases,  the  per-video throughput  vanishes as $1/K$.
Recently it was shown that \gp{scalable} per-video throughput can be achieved if the communications are synergistically designed with caching at the receivers. 
Indeed, the recent  breakthrough of \emph{coded caching} \cite{maddah2013fundamental} has inspired a  rethinking of  wireless downlink. 
Different video sub-files are cached at the receivers, \gp{and video requests are served by coded multicasts.} 

By careful selection of sub-file caching \gp{and exploitation of the  broadcast  wireless channel},   the transmitted signal is simultaneously
useful for \gp{decoding at users with  different  video requests}. 
Although this \gp{scheme--theoretically proved to scale well--can potentially} resolve the future downlink bottleneck, several limitations hinder its applicability in practical systems \cite{misconceptions}. In this work, we take a closer look to the limitations that arise from the fact that \emph{coded caching was originally designed for a symmetric error-free shared link.} 

\gp{If instead we consider a realistic model for the wireless channel, 
} we observe that a naive application of coded caching faces a \emph{short-term} limitation: since the channel qualities of the users fluctuate over time and our transmissions need to reach all users, the transmissions need to be designed for  the worst channel quality. This is in stark contrast with standard  downlink techniques, like \emph{opportunistic scheduling} \cite{stolyar,Li05, knopp}, which serve the user with the best instantaneous channel quality. 
Thus, a first challenge is to discover a way to allow coded caching technique to opportunistically exploit the fading of the wireless channel.

\gp{Apart from the fast fading consideration,} there is also a \emph{long-term} limitation due to the network topology. The user locations might vary, which leads to consistently poor channel quality for the ill-positioned users. 
The classical coded caching scheme is designed to deliver 
\gp{equal video shares to all users,
which leads to ill-positioned users consuming most of the air time and hence driving the
overall system performance to low efficiency.} In the literature, this problem has been resolved by the use of fairness among user throughputs \cite{Li05}. 
By allowing poorly located users to receive less throughput than others, precious airtime is saved and the overall system performance is greatly increased. Since the sum throughput rate and equalitarian fairness are typically the two extreme cases, past works have proposed the use of alpha-fairness \cite{mowalrand} which allows  to select the coefficient $\alpha$ and drive the system to any desirable  tradeoff point in between of the two extremes. 
Previously, the alpha-fair objectives have been studied in the context of (i) multiple user activations \cite{stolyar}, (ii) multiple antennas \cite{caire_fairnessMIMO} and (iii) broadcast channels \cite{caire_fairnessBC}. However, here the fairness problem is further complicated by the interplay between scheduling and the coded caching operation. In particular, we wish to shed light into the following questions: \emph{what is the right user grouping and how we should design the codewords to achieve our fairness objective while adapting to changing channel quality?}

To address these questions, we study the content delivery over a realistic block-fading broadcast channel, where the channel quality varies across users and time. 
In this setting, 
\gp{we design a scheme that decouples transmissions from coding. In the transmission side, we select the multicast user set dynamically depending on the instantaneous channel quality and user urgency captured by queue lengths. In the coding side, we adapt the codeword construction of \cite{maddah2013decentralized} depending on how fast the transmission side serves each user set. Combining with an appropriate congestion controller, we show that this approach yields our alpha-fair objective.}
More specifically, our approaches and contributions are summarized below:
\begin{itemize}
\item[1)] We impose a novel queueing structure which decomposes the channel scheduling from the codeword construction. Although it is clear that the codeword construction needs to be adaptive to channel variation, our scheme ensures this through our \emph{backpressure} that connects the user queues and the codeword queues. Hence, we are able to show that this decomposition is without loss of optimality.
\item[2)] We then provide an online policy consisting of (i) admission control of new files into the system; (ii) combination of files to perform coded caching; (iii) scheduling and power control of codeword transmissions to subset of users on the wireless channel. We prove that 
the long-term video delivery rate vector achieved by our scheme is a near optimal solution to the alpha-fair optimization problem under the specific coded caching scheme \cite{maddah2013decentralized}. 
\item[3)] Through numerical examples, we demonstrate the superiority of our approach versus (a) opportunistic scheduling with unicast transmissions and classical network caching (storing a fraction of each video), (b) standard coded caching based on transmitting-to-all.
\end{itemize}

\subsection{Related work}

Since coded caching was first proposed \cite{maddah2013fundamental} and its potential was recognized by the community, 
substantial efforts have been devoted to quantify the gain in realistic scenarios, including decentralized
placement \cite{maddah2013decentralized}, non-uniform popularities \cite{ji2015order,niesen2013coded}, and device-to-device (D2D)  networks \cite{ji2013fundamental}.
A number of recent works replace the original perfect shared link with wireless channels \cite{zhang2016wireless,bidokhti2016noisy,NgoAllerton2016}.
Commonly in the works with wireless channels, 
the performance of coded caching is limited by the user in the worst channel condition because the wireless multicast capacity is determined by the worst user \cite[Chapter 7.2]{el2011network}. \gp{This limitation of coded caching has been recently highlighted in  \cite{NgoAllerton2016}, while similar conclusions and some directions are given in  \cite{zhang2016wireless,bidokhti2016noisy}. Our work is the first to addresses this aspect by jointly designing the transmissions over the broadcast channel and scheduling appropriate subsets of users.}

Most past works deal with {\it offline} caching in the sense that both cache placement and delivery phases are performed once and do not capture the random and asynchronous nature of video traffic. 
The papers \cite{pedarsani2016online, niesen2015coded} addressed partly the online nature by studying cache eviction strategies, and delay aspects. 
In this paper, we explore a different online aspect. Requests for video files arrive in an online fashion, and transmissions are scheduled over time-varying wireless channels.

Online transmission scheduling over wireless channels has been extensively studied in the context of opportunistic scheduling \cite{stolyar} and network utility maximization \cite{neely10}. Prior works emphasize two fundamental aspects: (a) the balancing of user rates according to  fairness and efficiency considerations, and (b) the opportunistic exploitation of the time-varying fading channels.
Related to our work are the studies of wireless downlink with broadcast degraded channels; \cite{Seong06} gives a maxweight-type of policy and \cite{Eryilmaz01} provides a throughput optimal policy based on a fluid limit analysis.
Our work is the first to our knowledge that studies coded caching in this setting.
The new element in our study is the joint consideration of user scheduling with codeword construction for the coded caching delivery phase.

\section{System Model and Problem Formulation}
We study a wireless downlink consisting of a base station and $K$ users. 
The users are interested in downloading files over the wireless channel.

\subsection{Fair file delivery}

The performance metric is the \emph{time average delivery rate of files} to user $k$, denoted by $\overline{r}_k$. Hence our objective is expressed with respect to the vector of delivery rates $\rv$. 
We are interested in the \emph{fair file delivery} problem:
\begin{align}\label{eq:problem}
\overline{\boldsymbol{r}}^* = &\arg\max_{\overline{\boldsymbol{r}} \in \Lambda}\sum_{k=1}^Kg(\overline{r}_k),
\end{align}

where $\Lambda$ denotes the set of all feasible delivery rate vectors--clarified in the following subsection--and the utility function corresponds to the \emph{alpha fair} family of  concave functions obtained by choosing: 
\begin{align}
g(x) = \begin{cases}
\frac{(d+x)^{1-\alpha}}{1-\alpha}, \alpha\neq 1\\
\log(1+x/d), \alpha = 1
\end{cases}
\end{align}

for some arbitrarily small $d>0$ (used to extend the domain of the functions to $x=0$).  Tuning the value of $\alpha$ changes the shape of the utility function and consequently drives the system performance $\overline{\boldsymbol{r}}^*$ to different  points: (i) $\alpha=0$ yields max sum delivery rate, (ii)  $\alpha\to\infty$ yields max-min delivery rate \cite{mowalrand}, (iii)   $\alpha = 1$ yields proportionally fair delivery rate \cite{pfscheduling}.
Choosing $\alpha\in (0,1)$ leads to a tradeoff between max sum and proportionally fair delivery rates.

The optimization  \eqref{eq:problem}  is designed to allow us tweak the performance of the system;  we highlight its  importance by an example.
Suppose that for a 2-user system $\Lambda$ is given by the convex set shown on figure \ref{fig:example}.

Different boundary points are obtained as solutions to \eqref{eq:problem}.
If we choose $\alpha=0$, the system is operated at the point that maximizes the sum $\overline{r}_1+\overline{r}_2$. The choice $\alpha\to\infty$ leads to the maximum $r$ such that $\overline{r}_1=\overline{r}_2=r$, while $\alpha=1$ maximizes the sum of logarithms. 
The  operation point A is obtained when we always broadcast to all users at the weakest user rate and use \cite{maddah2013fundamental} for coded caching transmissions.
Note that this results in a significant loss of efficiency due to the variations of the fading channel, and consequently A lies in the interior of $\Lambda$.
We may infer that the point $\alpha\to\infty$ is obtained by avoiding transmissions to users with  instantaneous poor channel quality but still balancing their throughputs in the long run.
\begin{figure}
\vspace{-10pt}
\begin{center}
\includegraphics[width=0.25\textwidth,clip=]{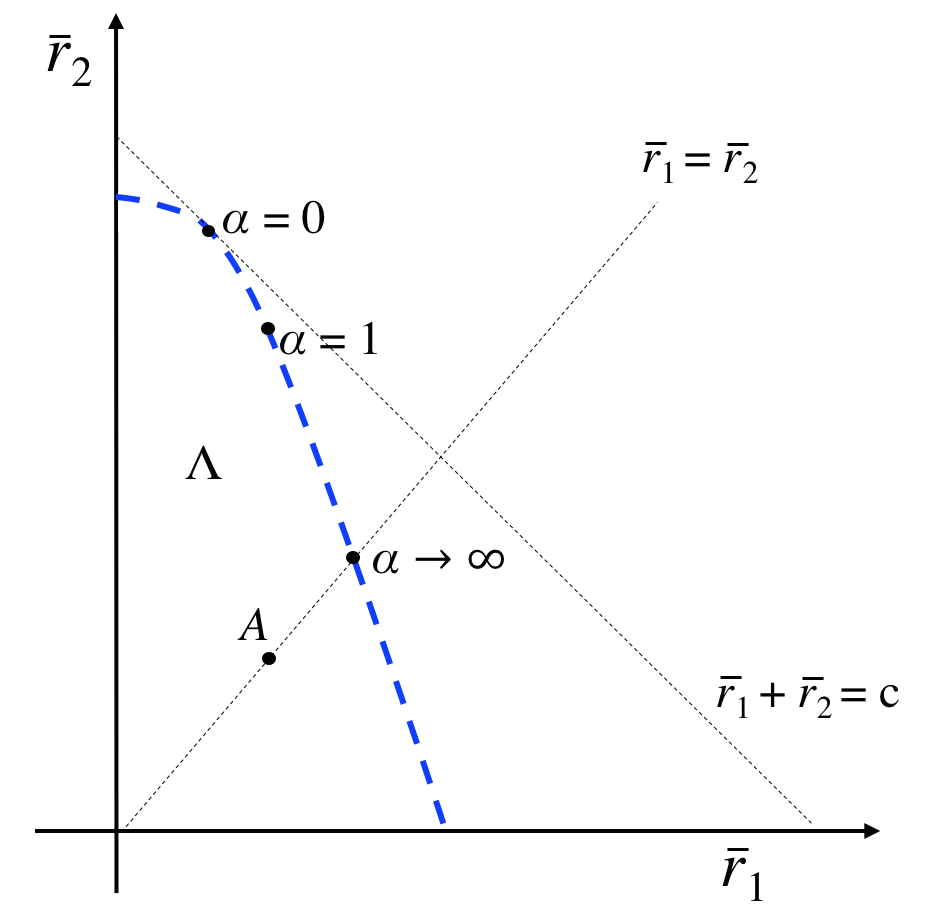}
\vspace{-10pt}
\caption{Illustration of the feasibility region and different performance operating points for $K=2$ users. Point A corresponds to a naive adaptation of \cite{maddah2013fundamental} on our channel model, while the rest points are solutions to our fair file delivery problem. }
\label{fig:example}
\end{center}
\vspace{-15pt}
\end{figure}

\subsection{Transmission model}

To analyze the set of feasible rate vectors $\Lambda$  we need to zoom in the detailed model of transmissions.

\textbf{Caching model.}
There are $N$ equally popular files $W_1, \dots, W_N$, each $F$ bits long.
The files are available to the base station.  User $k$ is equipped with cache memory $Z_k$ of $MF$ bits, where  $M\in [0,N]$.
Caching placement is performed during off-peak hour, and the goal is to fill the caches up to the memory constraint with selected bits.  To this end, we need to select $K$ \emph{caching functions} $\phi_k: \FF_2^{NF} \to\FF_2^{MF}$ which map the files $W_1, \dots, W_N$ into the cache contents 
\begin{align*}
Z_k \defeq \phi_k(W_1,\dots, W_N), ~~\forall k=1,\dots,K.
\end{align*}
The caching functions can be used to cache a few entire files, or a small fraction from each file, or even  coded combinations of subfiles \cite{maddah2013fundamental,pedarsani2016online}.
It is important to note that the caching functions are selected once, without knowledge of future requests, and are  fixed throughout our system operation.\footnote{A reasonable  extension is to enable  infrequent updates of the caching placement phase.}

\textbf{Downlink channel model.} We consider a standard block-fading broadcast channel, such that the channel state remains constant over a slot of $T_{\rm slot}$ channel uses  and changes from one slot to another in an i.i.d. manner. 
The channel output of user $k$ in any channel use of slot $t$ is given by  
\begin{align}\label{eq:BFC}
\yv_k(t)=\sqrt{h_k(t)} \xv(t)+\nuv_k(t),
\end{align}
where the channel input $\xv\in\CC^{T_{\rm slot}}$ is subject
to the power constraint $\EE[\Vert \xv \Vert^2] \leq PT_{\rm slot}$; $\nuv_{k}(t)\sim\Nc_{\CC}(0, \Id_{T_{\rm slot}})$ are additive white Gaussian noises with covariance matrix identity of size $T_{\rm slot}$, assumed independent of each other; $\{h_k(t)\in\CC\}$ are channel fading coefficients $\sim \beta_k^2\text{exp}(1)$ independently distributed across time and users, with $\beta_k$ denoting the path-loss parameter of user $k$. 

\textbf{Encoding and transmissions.}  The transmissions aim to contribute information towards the delivery of a specific vector of file requests $\dv(t)$, where $d_k(t)\in \{1, \dots, N\}$ denotes 
the index of the requested file by user $k$ in slot $t$. Here $N$ is the video library size, typically in the order of 10K. The requests are generated randomly, and whenever a file is delivered to user $k$, the next request of this user will be for another randomly selected file. 

At each time slot, the base station observes the channel state $\hv(t)=(h_1(t), \dots,h_K(t))$  and the request vector up to $t$, $\dv^t$, 
constructs a transmit symbol using the \emph{encoding function} $f_t:\{1,..,N\}^{Kt} \times \CC^{K} \to\CC^{T_{\rm slot}}$. 
 \[
 \xv(t) = f_t \left( \dv^t,\hv(t)\right),
 \]
 
  Finally, it transmits a codeword $\xv(t)$ for the $T_{\rm slot}$ channel uses  over the fading broadcast channel in slot $t$ . 
  The encoding function may be chosen at each slot  to contribute information to a selected subset of users $\mathcal{J}(t)\subseteq \{1,\dots,K\}$. This allows several possibilities, e.g. to send more information to a small set of users with good instantaneous channel qualities, or less information to a large set that includes users with poor quality. 

\textbf{Decoding.}
At slot $t$, each user $k$ observes the local cache contents $Z_k$ and the sequence of channel outputs so far $ y_k(\tau),~~\tau=1,\dots,t$ and employs a \emph{decoding function} $\xi_k$ to determine the decoded files.
Let $D_k(t)$ denote the number of files decoded by user $k$ after $t$ slots.
The decoding function $\xi_k$ is a mapping

  \[
\xi_k: \CC^{T_{slot}t}\times\CC^{Kt} \times \FF_2^{FM}  \times \{1,..,N\}^{Kt} \to  \FF_2^{FD_k(t)}.
  \]
The decoded files of user $k$ at slot $t$ are given by $\xi_k(y_k^{T_{\rm slot}t}, Z_k, \hv^t, \dv^t)$,
and depend on the channel outputs and states up to $t$, the local cache contents, and the requested files of all users up to $t$. 
A file is incorrectly decoded if it does not belongs to the set of requested files. The number of incorrectly decoded files are then given by $|\cup_t\{\xi_k(t)\} \setminus d_k^t|$ and
the number of correctly decoded files at time $t$ is: 
\[
C_k(t)= D_k(t)- |\cup_t\{\xi_k(t)\} \setminus d_k^t|
\]

\begin{definition}[Feasible rate]
A rate vector $\overline{\rv}=(\overline{r}_1,\dots, \overline{r}_K)$ is said to be {\it feasible} $\overline{\rv}\in\Lambda$ if there exist functions $([\phi_k], [f_t], [\xi_k])$  such that:

\vspace{-0.05in}
\[
\overline{{r}}_k=\limsup_{t\rightarrow\infty} \frac{C_k(t)}{t},
\]

where the rate is measured in file/slot. 
\end{definition}

In contrast to past works which study the performance of one-shot coded caching \cite{maddah2013fundamental,maddah2013decentralized,pedarsani2016online}, our rate metric measures the ability of the system to continuously deliver files to users.

\subsection{Code-constrained rate region}
{Finding the optimal policy is very complex. 
In this paper, we restrict the problem to specific class of policies given by the following mild assumptions: 
\begin{definition}[Admissible class policies $\Pi^{CC}$]
The admissible policies have the following characteristics: 
\begin{enumerate}
\item The caching placement and delivery follow the decentralized scheme \cite{maddah2013decentralized}.
\item The users request distinct files, i.e., the ids of the requested files of any two users are different.
\end{enumerate}
\end{definition}
Since we restrict our action space, the delivery rate feasibility region, $\Lambda^{CC}$, of the class of policies $\Pi^{CC}$ is  smaller than the one for the original problem $\Lambda$. However, these restrictions allow us to come up with a concrete solution approach. Note that the optimal cache and transmission design policy is already a very hard problem even in the simple case of broadcast transmissions with a fixed common rate, and the method in \cite{maddah2013fundamental},   \cite{maddah2013decentralized}  are  practical approaches with good performance.  In addition, looking at demand IDs when combining files would be very complex and, because of the big library sizes, is not expected to bring substantial gains (it is improbable that two users will make request for the same file in close time instances).   
}

\section{Offline Coded Caching}

In this section we briefly review decentralized coded caching, first proposed in \cite{maddah2013decentralized}, and used by all  admissible policies $\Pi^{CC}$.
We set $m=\frac{M}{N}$ the normalized memory size. Under the memory constraint of $MF$ bits, each user $k$ independently caches a subset of $mF$ bits of file $i$, chosen uniformly at random for $i=1,\dots, N$. By letting $W_{i|\Jc}$ denote the sub-file of $W_i$ stored exclusively in the cache memories of the user set $\Jc$, the cache memory $Z_k$ of user $k$ after decentralized placement is given by
\begin{align} \label{eq:Zk}
Z_k =\{ W_{i \cond \Jc}:  \;\; \forall \Jc \subseteq[K], \forall \Jc \ni k , \forall i =1,\dots, N \}.
\end{align}
The size of each sub-file measured in bits is given by  
\begin{align}
|W_{i \cond \Jc}|= m^{|\Jc|}\left(1-m\right)^{K-|\Jc|} \label{eq:LLN}
\end{align}
as $F\rightarrow \infty$. The above completely determine the caching functions.

Once the requests of all users are revealed, the offline scheme proceeds to the delivery of the requested files (delivery phase). Assuming that user $k$ requests file $k$, i.e. $d_k =k$, the server generates and conveys the following codeword simultaneously useful to the subset of users $\Jc$:  
\begin{align}
V_{\Jc}=\oplus_{k\in \Jc}W_{k|\Jc\setminus\{k\}},
\end{align}
where $\oplus$ denotes the bit-wise XOR operation. The main idea here is to create a codeword useful to a subset of users by exploiting the receiver side information established during the placement phase. 
It is worth noticing that the {\it coded} delivery with XORs significantly reduces the number of transmissions. 
Compared to uncoded delivery, where the sub-files are sent sequentially and the number of transmissions are equal to $|\Jc|\times|W_{k|\Jc\setminus\{k\}}|$, the coded delivery requires the transmission of $|W_{k|\Jc\setminus\{k\}}|$, yielding a reduction of a factor $|\Jc|$.
In a practical case of $N>K$, it has been proved that decentralized coded caching achieves the total number of transmissions, measured in the number of files, given by \cite{maddah2013decentralized}
\begin{align}\label{eq:mad13}
T_{\rm tot}(K,m) = \frac{1}{m} \left(1-m\right) \left\{1-\left(1-m\right) ^K \right \}. 
\end{align}
On the other hand, in uncoded delivery, the number of transmissions is given by $K(1-m)$ since it exploits only {\it local} caching gain at each user. For a system with $K=30$ users and normalized memory of $m=1/3$, the minimum transmissions required by uncoded delivery is $20$ and that of decentralized coded caching is $2$, yielding a gain of factor $10$.  

 In order to further illustrate the placement and delivery of decentralized coded caching,  we provide an three-user example. 
\begin{example}
For the case of $K=3$ users in Fig.\ref{fig:ex3}, let us assume that user 1, 2, 3, requests file $A, B, C$, respectively.  
After the placement phase, a given file $A$ will be partitioned into 8 subfiles.

\begin{figure}
\vspace{-8pt}
\begin{center}
\includegraphics[width=0.4\textwidth,clip=]{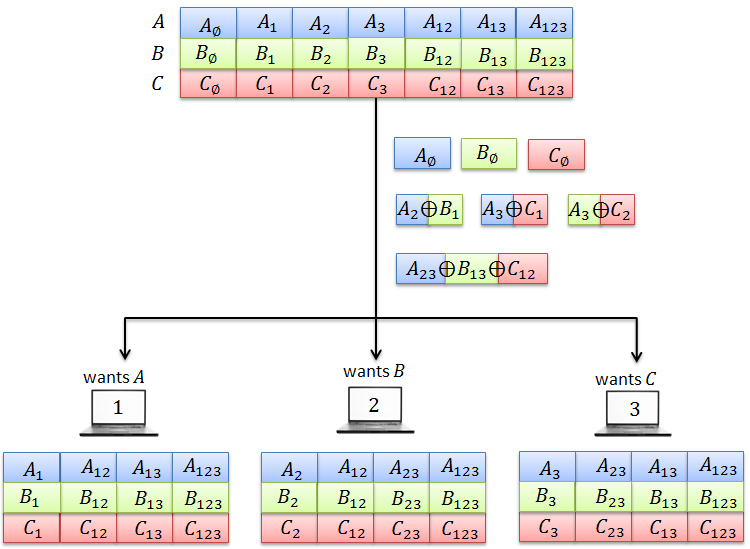}
\vspace{-2pt}
\caption{Decentralized coded caching for $K=3$}
\label{fig:ex3}
\end{center}
\vspace{-20pt}
\end{figure}

Codewords to be sent are the following
\begin{itemize}
\item  $A_{ \emptyset}$, $B_{ \emptyset}$ and $C_{ \emptyset}$ to user $1$, $2$ and $3$ respectively.  
\item  $A_{ 2}\oplus B_{ 1}$ is intended to users $\{1,2\}$. Once received, user $1$ decodes $A_{ 2}$ by combining the received codeword with $B_{ 1}$ given in its cache. Similarly user $2$ decodes $B_{ 1}$. The same approach holds for codeword $B_{ 3}\oplus C_{ 2}$ to users $\{2,3\}$ and codeword $A_{ 3}\oplus C_{ 1}$ to users $\{1,3\}$
\item  $A_{ 23}\oplus B_{ 13}\oplus C_{ 12}$ is intended users ${1,2,3}$. User $1$ can decode $A_{ 23}$ by combining the received codeword with $\{B_{ 13},C_{ 12}\}$ given in its cache. The same approach is used for user $2$, $3$ to decode $B_{ 13}$, $C_{ 12}$ respectively.  
\end{itemize}     
\end{example}

\section{Broadcasting Private and Common Messages}\label{section:broadcasting}

In this section, we address the question on how the transmitter shall convey private and multiple common messages, each intended to a subset of users, while opportunistically exploiting the underlying wireless channel.  We start by remarking that the channel in \eqref{eq:BFC} for a given channel realization $\hv$ corresponds to the Gaussian degraded broadcast channel. Without loss of generality, let us assume $h_1\geq \dots \geq h_K$ so that the following Markov chain holds.
\vspace{-0.05in}
 \[X \leftrightarrow Y_1  \leftrightarrow  \dots  \leftrightarrow  Y_K. \]
\vspace{-0.05in}
The capacity region of the degraded broadcast channel for $K$ private messages and a common message is well-known \cite{el2011network}.
 In this section, we consider a more general setup where the transmitter wishes to convey 
$2^K-1$ mutually independent messages, denoted by $\{M_{\Jc}\}$, where $M_{\Jc}$ denotes the message intended to the users in subset $\Jc\subseteq \{1,\dots, K\}$. Each user $k$ must decode all messages $\{M_{\Jc}\}$ for $\Jc\ni k$. By letting $R_{\Jc}$ denote the multicast rate of the message $M_{\Jc}$, we say that the rate-tuple $\Rm\in \RR_+^{2^K-1}$ is achievable if there exists encoding and decoding functions which ensure the reliability and the rate condition. The capacity region is defined as the supremum of the achievable rate-tuple, where the rate is measured in bit/channel use. 
\begin{theorem}\label{th:capacity_BC}
The capacity region $\Gamma(\hv)$ of a $K$-user degraded Gaussian broadcast channel with fading gains $h_1 \geq \dots \geq h_K$ and $2^K-1$ independent messages $\{M_{\Jc}\}$ is given by 
\begin{align}
R_1 & \leq \log(1+ h_1\alpha_1 P) \\ \label{eq:9}
\sum_{\Jc \subseteq \{1,\dots, k\}: k\in\Jc} R_{\Jc} & \leq \log\frac{1+ h_k \sum_{j=1}^{k} \alpha_j P}{1+ h_k\sum_{j=1}^{k-1} \alpha_j P} \;\;\; k=2, \dots, K\end{align}
for non-negative variables $\{\alpha_k\}$ such that $\sum_{k=1}^K \alpha_k \leq 1$. 
\end{theorem}
\begin{proof}
 Please refer to Appendix \ref{appendix:BC_capacity_proof} for the proof.
\end{proof}

The achievability builds on superposition coding at the transmitter and successive interference cancellation at receivers. For $K=3$, the transmit signal is simply given by 
\[x = x_1 + x_2 + x_3 + x_{12} + x_{13}+ x_{123}
\]
where $\{x_{\Jc}\}$ are mutually independent Gaussian distributed random variables satisfying the power constraint and $x_{\Jc}$ denotes the signal corresponding to the message $M_{\Jc}$ intended to the subset $\Jc\subseteq \{1,2,3\}$.   
User 3 (the weakest user) decodes $\tilde{M}_3 =\{M_{3}, M_{13}, M_{23}, M_{123}\}$ by treating all the other messages as noise. User 2 decodes first the messages $\tilde{M}_3$ and then jointly decodes $\tilde{M}_2 =\{M_2, M_{12}\}$. Finally, user 1 (the strongest user) successively decodes $\tilde{M}_3, \tilde{M}_2$ and, finally, $M_1$. 

Later in our online coded caching scheme we will need the capacity region $\Gamma(\hv)$, and more specifically, we will need to characterize its boundary. 
To this end,   

it suffices to  consider the weighted sum rate maximization:
\begin{align}\label{eq:generalWSR}
\max_{\rv \in \Gamma(\hv)}\sum_{\Jc: \Jc \subseteq \{1,\dots, K\}} \theta_{\Jc} r_{\Jc}.
\end{align}
We first simplify the problem using the following theorem.

\begin{theorem}
The weighted sum rate maximization with $2^K-1$ variables in \eqref{eq:generalWSR} reduces to a simpler problem with $K$ variables, given by 
\begin{align}\label{eq:powerallocation}
f(\alphav) = \sum_{k=1}^K \tilde{\theta}_k \log\frac{1+h_k \sum_{j=1}^{k} \alpha_j P}{1+ h_k\sum_{j=1}^{k-1} \alpha_j P}. 
\end{align}
where $\tilde{\theta}_k$ denotes the largest weight for user $k$ 
\[
\tilde{\theta}_k=\max_{\Kc: k\in \Kc \subseteq \{1,\dots, k\}}\theta_{\Kc}.
\]
\end{theorem}
\begin{proof}
The proof builds on the simple structure of the capacity region. We first remark that for a given power allocation of other users, user $k$ sees $2^{k-1}$ messages $\{W_{\Jc}\}$ for all $\Jc$ such that $k\in \Jc \subseteq \{1,\dots, k\}$ with the equal channel gain. For a given set of $\{\alpha_j\}_{j=1}^{k-1}$, the capacity region of these messages is a simple hyperplane 
characterized by $2^{k-1}$ vertices $C_k \ev_i$ for $i=1, \dots, 2^{k-1}$, where $C_k$ is the sum rate of user $k$ in the RHS of \eqref{eq:9} and $\ev_i$ is a vector with one for the $i$-th entry and zero for the others. Therefore, the weighted sum rate seen is maximized for user $k$ by selecting the vertex corresponding to the largest weight, denoted by $\tilde{\theta}$. This holds for any $k$. 
\end{proof}

We provide an efficient algorithm to solve this power allocation problem as a special case of the parallel Gaussian broadcast channel studied in \cite[Theorem 3.2]{TseOptimal}. Following \cite{TseOptimal}, we define the rate utility function for user $k$ given by 
\begin{align}
u_k(z)= \frac{\tilde{\theta}_k}{1/h_k+z}-\lambda
\end{align}
where $\lambda$ is a Lagrangian multiplier. The optimal solution corresponds to selecting 
the user with the maximum rate utility at each $z$ and the resulting power allocation for user $k$ is 
\begin{align}\label{eq:optimalalpha}
 \alpha^*_k = \left\{ z:   [\max_j u_j(z) ]_+ = u_k(z)  \right \}/P
\end{align}
with $\lambda$ satisfying 
\begin{align} \label{eq:195}
P=\left [ \max_k \frac{\tilde{\theta}_k}{\lambda} -\frac{1}{h_k} \right]_+.
\end{align}

\section{Proposed Online Delivery Scheme}
This section presents first the queued delivery network and its feasible rate region of arrival rates, then describes the proposed
control policy. 

\subsection{Solution plan}

At each time slot $t$, the controller \emph{admits} $a_k(t)$ files to be delivered to user $k$, and hence $a_k(t)$ is a control variable.\footnote{We note that random file arrivals can be directly captured  with the addition of an extra queue \cite{neely10}, which  we avoid  to simplify exposition.}
As our model dictates, the  succession of requested files for user $k$ is determined uniformly at random.

\textbf{Queueing model.}
The base station organizes the information into the following types of queues: 
\begin{enumerate}
	\item \textbf{User queues}  to store admitted files, one for each user. The buffer size of queue $k$ is denoted by $S_k(t)$ and expressed in number of files. 
	\item \textbf{Codeword queues} to store codewords to be multicast. There is one codeword queue for each subset of users $\Jc \subseteq \{1, \dots, K\}$. The size of codeword queue $\Jc$ is denoted by $Q_{\Jc}(t)$ and expressed in bits.
\end{enumerate} 
A queueing policy $\pi$ performs the following operations: (i) decides how many files to admit into  the user queues $S_k(t)$ in the form of $(a_k(t))$ variables, (ii) then  it decides how to combine together files from different user queues to be encoded into the form of multiple codewords which represent the required broadcast transmissions for the reception of this file--these codewords are stored in the appropriate codeword queues $Q_{\Jc}(t)$, (iii) and last it decides the encoding function $f_t$.
(ii) and (iii) are further clarified in the next section.

\begin{definition}[Stability]
	A queue $S(t)$ is said to be (strongly) stable if
	\vspace{-0.1in}
	\[
	\limsup\limits_{T\rightarrow\infty}\frac{1}{T}\sum_{t=0}^{T-1}\mathbb{E}\left[S(t)\right] < \infty.
	\]
	A queueing system is said to be stable if all its queues are stable. Moreover, the stability region of a system is the set of all arrival rates such that the system is stable.   
\end{definition}
The above definition implies that the average delay of each job in the queue is finite. 

In our problem,  if we develop a policy that keeps {\it user queues} $\Sm(t)$ stable, then all admitted files will, at some point, be combined into codewords. If in addition {\it codeword queues} $\Qm(t)$ are stable, then all generated codewords will reach their destinations, meaning that all receivers will be able to decode the admitted files that they requested. 
\begin{lemma}\label{lem:equivalence}
The region of all feasible delivery rates $\Lambda$ is the same as the stability region of the system (i.e. the set of all demand arrival rates for which there exists a policy that stabilizes the queueing system). 
\end{lemma}

Let $\overline{a}_k = \limsup\limits_{t\rightarrow\infty}\frac{1}{t}\sum_{t=0}^{t-1}\mathbb{E}\left[a_k(t)\right],$
denote the time average number of admitted files for user $k$.  Lemma \ref{lem:equivalence} implies the following Corollary. 
\begin{corollary}\label{cor:equivalentOptimization}
Solving  \eqref{eq:problem} is equivalent to finding a policy $\pi$ such that 
\vspace{-0.2in}	
\begin{align}\label{eq:objective2}
\overline{\boldsymbol{a}}^{\pi} =& \argmax\sum_{k=1}^Kg_k(\overline{a}_k)\\ 
\text{s.t.}\quad & \text{the system is stable.} \notag
\end{align}
\end{corollary}
\begin{proof}
See Appendix \ref{appendix:equivalence_proof}
\end{proof}
	
\subsection{Feasible Region}
Contrary to the offline coded caching in \cite{maddah2013decentralized}, we propose an online delivery scheme consisting of the following three blocks. Each block is operated at each slot. 
\begin{enumerate}

\item \textbf{Admission control:}
At the beginning of each slot, the controller decides how many requests for each user, $a_k(t)$ should be pulled into the system from the infinite reservoir.

\item \textbf{Routing:}
The cumulative accepted files for user $k$ are stored in the admitted demand queue whose size is given by $S_k(t)$ for $k=1,\dots, K$. The server decides the combinations of files to perform coded caching. The decision at slot $t$ for a subset of users $\Jc\subseteq\{1,..,K\}$, denoted by $\sigma_{\Jc}(t)\in\{0,1,\dots,\sigma_{\max}\}$, refers to the number of combined requests for this subset of users. It is worth noticing that offline coded caching lets $\sigma_{\Jc} = 1$ for $\Jc =\{1,\dots, K\}$ and zero for all the other subsets.  
The size of the queue $S_k$ evolves as: 
\begin{align}
S_k(t+1) = \left[S_k(t) - \sum_{\Jc: k\in \Jc}\sigma_{\Jc}(t)\right]^+ + a_k(t)  \label{eq:codewordQueues}
\end{align}
If $\sigma_{\Jc}(t)> 0$, the server creates codewords by applying offline coded caching explained in Section [{\bf}] for this subset of users 
as a function of the cache contents $\{Z_j: j\in \Jc\}$. 
  
\item \textbf{Scheduling:}
  The codewords intended to the subset $\Jc$ of users are stored in codeword queue whose size is given by $Q_{\Ic}(t)$ for $\Ic\subseteq \{1,\dots, K\}$. Given the instantaneous channel realization $\hv(t)$ and the queue state $\{Q_{\Ic}(t)\}$, the server performs scheduling and rate allocation. Namely, at slot $t$, it determines the number $\mu_{\Ic}(t)$ of bits per channel use to be transmitted for the users in subset $\Ic$. By letting $b_{\Jc,\Ic}$ denote the number of bits generated for codeword queue $\Ic\subseteq \Jc$ when offline coded caching
  is performed to the users in $\Jc$, codeword queue $\Ic$ evolves as 
   \begin{align}\nonumber
   Q_{\Ic}(t+1) =  \left[Q_{\Ic}(t) - T_{\rm slot}\mu_{\Ic}(t)\right]^+  + \sum_{\Jc:\Ic\subseteq\Jc}b_{\Jc,\Ic}\sigma_{\Jc}(t)
   \end{align}
   where $b_{\Jc,\Ic}=m^{|\Ic|}(1-m)^{|\Jc|-|\Ic|-1}$.
 \end{enumerate}  

In order determine our proposed policy, namely the set of decisions $\{\av(t), \sigmav(t),\muv(t)\}$ at each slot $t$,  we first characterize the feasible region $\Lambda$ as a set of arrival rates $\av$. We let $\pi_{\mathbf{h}}$ denote the probability that the channel state at slot $t$ is $\hv\in \Hc$ where $\Hc$ is the set of all possible channel states. We let $\Gamma(\hv)$ denote the capacity region for a fixed channel state $\mathbf{h}$. Then we have the following 
\begin{theorem}[Feasibility region $\Lambda^{CC}$]\label{th:feasibilityRegion}
	A demand rate vector is feasible, i.e. $\bar{\av} \in \Lambda^{CC}$, if and only if there exist $\muv \in \sum_{\mathbf{h}\in\mathcal{H}}\pi_{\mathbf{h}}\Gamma(\mathbf{h})$, $\bar{\sigma}_{\Ic} \in[0,\sigma_{max}], \forall \Ic\subseteq \{1,\dots, K\}$ such that: 
	\begin{align}\label{eq:all_packets_get_combined}
	\sum_{\Jc: k\in \Jc}\bar{\sigma}_{\Jc} \geq \bar{a}_k, \forall k =1,\dots,K \\ \label{eq:all_codewords_get_transmitted} 
	T_{slot}\mu_{\Ic} \geq \sum_{\Jc: \Ic \subseteq \Jc} b_{\Jc, \Ic}\bar{\sigma}_{\Jc}, 
	\forall \Ic\in 2^{\Kc}
.	\end{align}
\end{theorem}
Constraint \eqref{eq:all_packets_get_combined} says that the service rate at which admitted demands are combined to form codewords is greater than the arrival rate, while \eqref{eq:all_codewords_get_transmitted} implies that the long-term average transmission rate $\overline{\mu}_{\Ic}$ for the subset $\Ic$ of users should be higher than the rate at which bits of generated codewords for this group arrive. In terms of the queueing system defined, these constraints impose that the service rates of each queue should be greater than their arrival rates, thus rendering them stable.   

Theorem \ref{th:feasibilityRegion} implies that the set of feasible average delivery rates is a convex set. 

\subsection{Admission Control and Routing}

In order to perform the utility maximization \eqref{eq:objective2}, we need to introduce one more set of queues. These queues are virtual, in the sense that they do not hold actual file demands or bits, but are merely counters to drive the control policy. Each user $k$ is associated with a queue $U_k(t)$ which evolves as follows: 
\begin{align}
U_k(t+1) = \left[U_k(t) - a_k(t)\right]^+ + \gamma_k(t)
\end{align}

where $\gamma_k(t)$ represents the arrival process to the virtual queue and is given by 
\begin{align}
	\gamma_k(t) = \arg\max\limits_{0\leq x\leq \gamma_{k,max}}\left[Vg_k(x) - U_k(t)x\right]
\end{align}
In the above, $V>0$ is a parameter that controls the utility-delay tradeoff achieved by the algorithm (see Theorem \ref{th:optimality_infinite}). 

The general intuition here is as follows: Observe that the number $a_k(t)$ of admitted demands is the service rate for the virtual queues $U_k(t)$.  The control algorithm actually seeks to optimize the time average of the virtual arrivals $\gamma_k(t)$. However, since $U_k(t)$ is stable, its service rate, which is the actual admission rate, will be greater than the rate of the virtual arrivals, therefore giving the same optimizer. Stability of all other queues will guarantee that admitted files will be actually delivered to the users. 

We present our on-off policy for admission control and routing.  For every user $k$, admission control chooses $a_k(t)$ demands given by 
\begin{align}
	a_k(t) = \gamma_{k, max} \onev\{U_k(t) \geq S_k(t) \}
\end{align}

For every subset $\Jc \subseteq \{1,\dots, K\}$, routing combines $\sigma_{\Jc}(t)$ demands of users in $\Jc$ given by
\begin{align}
\sigma_{\Jc}(t) = \sigma_{max} \onev\left\{ \sum_{k\in\Jc}S_k(t) > \sum_{\Ic: \Ic \subseteq \Jc}\frac{b_{\Jc, \Ic}}{F^2} Q_{\Ic}(t)  \right\}
.\end{align}

\subsection{Scheduling and Transmission}
In order to stabilize all {\it codeword queues}, the scheduling and resource allocation explicitly solve the following weighted sum rate 
maximization at each slot $t$ where the weight of the subset $\Jc$ corresponds to the queue length of $Q_{\Jc}$ 
\begin{align}\label{eq:WSRstability}
\muv(t) = \arg\max\limits_{\rv \in\Gamma(\hv(t))}\sum_{\Jc\subseteq \{1, \dots, K\}} Q_{\Jc}(t)r_{\Jc} 
.\end{align}
We propose to apply the power allocation algorithm in Section \ref{section:broadcasting} to solve the above problem by sorting users in a decreasing order of channel gains and treating $Q_{\Jc}(t)$ as $\theta_{\Jc}$.  
In adition, we assume that the number of channel uses in one coherence block is large enough such that the decoding error from choosing channel  codes with rate $\muv(t)$ is very small. In this case, no feedback from the receivers is given. 

\subsection{Example} 
We conclude this section by providing an example of our proposed online delivery network for $K=3$ users as illustrated in Fig. \ref{fig:queueing_system_real}.
\begin{figure}
\vspace{-10pt}
	\centering
	\includegraphics[scale=0.46]{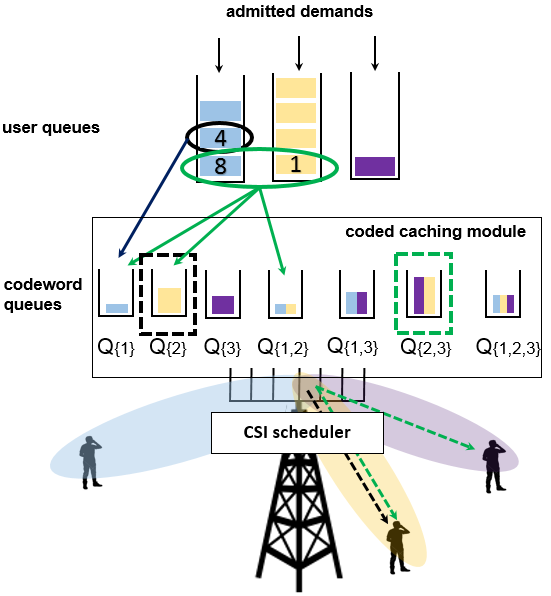}
	\vspace{-10pt}
	\caption{An example of the queueing model for a system with $3$ users. Dashed lines represent wireless transmissions, solid circles files to be combined and solid arrows codewords generated.}
	\label{fig:queueing_system_real}
	\vspace{-10pt}
\end{figure}
 At slot $t$ the server decides to combine $W_1$ requested by user $2$ with $W_8$ requested by user $2$ and  to process $W_4$ requested by user $1$ uncoded. Therefore $\sigma_{\{1,2\}}(t)=\sigma_{\{1\}}(t)=1$ and $\sigma_{\Jc}(t)=0$ otherwise. Given this codeword construction, codeword queues have inputs as described in Table I.
\begin{table}[ht]
\vspace{-10pt}
\caption{Codeword queues inputs.}
\vspace{-10pt}
\label{tab:1}
\begin{center}
\begin{tabular}{c|c}\hline
Queue & Input\\
\hline
$\Qc_{\{1\}}$ & $W_{8, \emptyset}$; $W_{8|3}$\\ 
              & $W_{4|\emptyset}$; $W_{4|\{2\}}$; $W_{4|\{3\}}$; $W_{4|\{2,3\}}$\\
\hline
$\Qc_{\{2\}}$ & $W_{1|\emptyset}$; $W_{1|\{3\}}$\\
\hline
$\Qc_{\{1,2\}}$ & $W_{1|\{1\}}\oplus W_{8|\{2\}}$; $W_{1|\{1,3\}}\oplus W_{8|\{2,3\}}$\\
\hline
\end{tabular}
\end{center}
\vspace{-10pt}
\end{table} 
In addition, data from queues $Q_{\{2\}}(t), Q_{\{2.3\}}(t)$ are transmitted.  

\section{Performance Analysis}
In thi section, we present the main result of the paper, that our proposed online algorithm leads to close to optimal performance for all policies in the class $\Pi^{CC}$:

\begin{theorem}\label{th:optimality_infinite}
	Let $\bar{r}^{\pi}_k$ the mean time-average delivery rate for user $k$  achieved by the proposed policy. Then 
	\begin{align}\nonumber
	\sum_{k=1}^Kg_k(\bar{r}^{\pi}_k) \geq \max_{\bar{\rv} \in\Lambda^{CC}}\sum_{k=1}^Kg_k(\bar{r}_k) - \bigO{\frac{1}{V}} \\ \nonumber
\limsup\limits_{T\rightarrow \infty}\frac{1}{T}\sum_{t=0}^{T-1}\mathbb{E}\left\{\hat{Q}(t)\right\}  = \bigO{V},
\end{align}
where $\hat{Q}(t)$ is the sum of all queue lengths at the beginning of time slot $t$, thus a measure of the mean delay of file delivery. 
\end{theorem}
The above theorem states that, by tuning the constant $V$, the utility resulting from our online policy can be arbitrarily close to the optimal one, where there is a tradeoff between the guaranteed optimality gap $\bigO{1/V}$ and the upper bound on the total buffer length $\bigO{V}$.

For proving the Theorem, we use the Lyapunov function 
\[
L(t) = \frac{1}{2}\left(\sum_{k=1}^KU_k^2(t) + S_k^2(t) + \sum_{\Ic \in 2^{\Kc}}\frac{1}{F^2}Q_{\Ic}^2(t)\right)
\]
and specifically the related drift-plus-penalty quantity, defined as: $\mathbb{E}\left\{L(t+1) - L(t)| \mathbf{S}(t), \mathbf{Q}(t), \mathbf{U}(t)\right\} - V\mathbb{E}\left\{\sum_{k=1}^Kg(\gamma_k(t))|\mathbf{S}(t), \mathbf{Q}(t), \mathbf{U}(t) \right\}$. The proposed algorithm is such that it minimizes (a bound on) this quantity. The main idea is to use this fact in order to compare the evolution of the drift-plus-penalty under our policy and two "static" policies, that is policies that take random actions (admissions, demand combinations and wireless transmissions), drawn from a specific distribution, based only on the channel realizations (and knowledge of the channel statistics). We can prove from Theorem 4 that these policies can attain every feasible delivery rate. The first static policy is one such that it achieves the stability of the system for  an arrival rate vector $\av'$ such that $\av'+\mathbf{\delta} \in \partial\Lambda^{CC}$. Comparing with our policy, we deduce strong stability of all queues and the bounds on the queue lengths by using a Foster-Lyapunov type of criterion. In order to prove near-optimality, we consider a static policy that admits file requests at rates $\av^* = \arg\max_{\av}\sum_kg_k(a_k)$ and keeps the queues stable in a weaker sense (since the arrival rate is now in the boundary $\Lambda^{CC}$). By comparing the drift-plus-penalty quantities and using telescopic sums and Jensen's inequality on the time average utilities, we obtain the near-optimality of out proposed policy.  

The full proof is in Appendix \ref{appendix:performance_proof}.

\section{Numerical Examples}

In this section, we compare our proposed delivery scheme with the following two other schemes, all building on decentralized cache placement in \eqref{eq:Zk} and \eqref{eq:LLN}. 

\begin{itemize}
\item \textbf{Unicast opportunistic
scheduling}: for any request, the server sends the remaining $(1-m)F$ bits to the corresponding user without combining any files. Here we only exploit the local caching gain. In each slot the serve sends with full power to   user
\vspace{-0.1in}
\[ k^{*}(t)=\arg\max_{k}\frac{\log\left( 1+h_k(t)P\right) }{ T_k(t) ^{\alpha}},
\]
 where $T_k(t)=\frac{\sum_{1\leq\tau\leq t-1} \mu_{k}(\tau)}{(t-1)}$ is the empirical average rate for user $k$ up to slot $t$. 

\item \textbf{Standard coded caching}: we use decentralized coded caching among all $K$ users. For the delivery, non-opportunistic TDMA transmission is used. The server sends sequentially codewords $V_{\Jc}$ to the subset of users $\Jc$ at the weakest user rate among $\Jc$: 
\vspace{-0.1in}
\[
\mu_{\Jc}(t)=\log\left( 1+P\min_{k\in\Jc}(h_k(t))\right) .
\]Once the server has sent codewords $\{V_{\Jc}\}_{\emptyset\neq\Jc\subseteq\{1,..,K\}}$, every user is able to decode one file. Then the process is repeated for all the demands. 
\end{itemize}

We consider the system with  normalized memory of $m=0.6$, power constraint $P=10dB$, file size $F=10^3$ bits and number of channel uses per slot $T_{\rm slot}=10^2$. We divide users into two classes of $K/2$ users each: strong users with $\beta_k=1$ and weak users with $\beta_k=0.2$.

We compare the three algorithms for the cases where the objective of the system is sum rate maximization ($\alpha=0$) and proportional fairness ($\alpha=1$). The results are depicted in Fig.~\ref{subfig-1:sumRates} and ~\ref{subfig-2:Utilities}, respectively.

   \begin{figure}
     \subfloat[\scriptsize Sum rate ($\alpha=0$)\label{subfig-1:sumRates}]{%
       \includegraphics[width=0.235\textwidth]{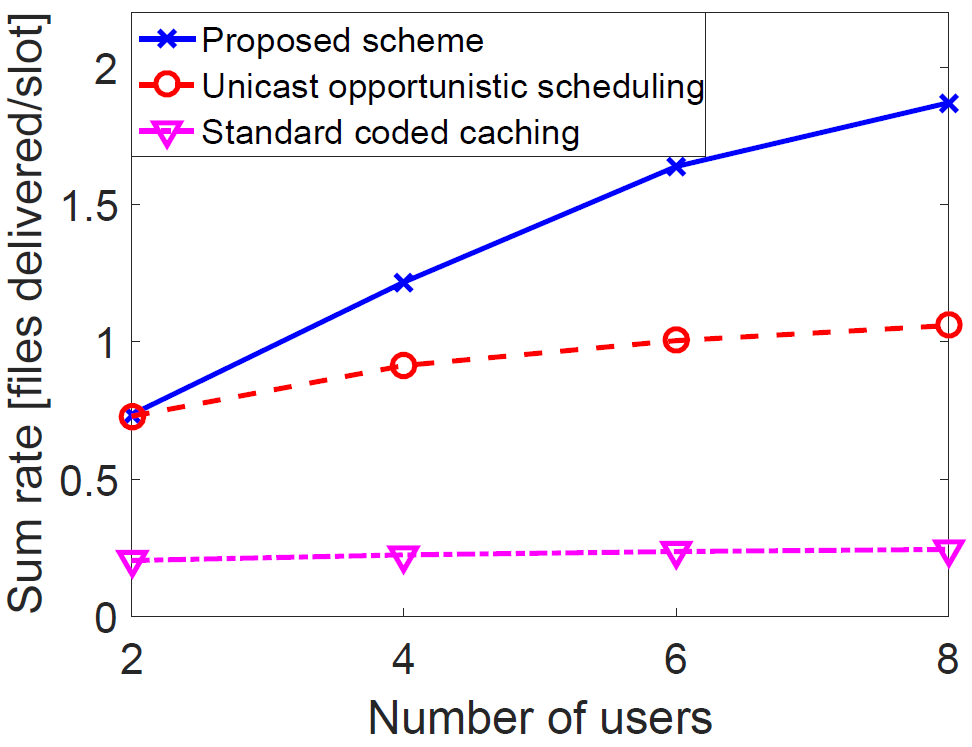}
     }
     \hfill
     \subfloat[\scriptsize Proportional fair utility ($\alpha=1$)\label{subfig-2:Utilities}]{%
       \includegraphics[width=0.235\textwidth]{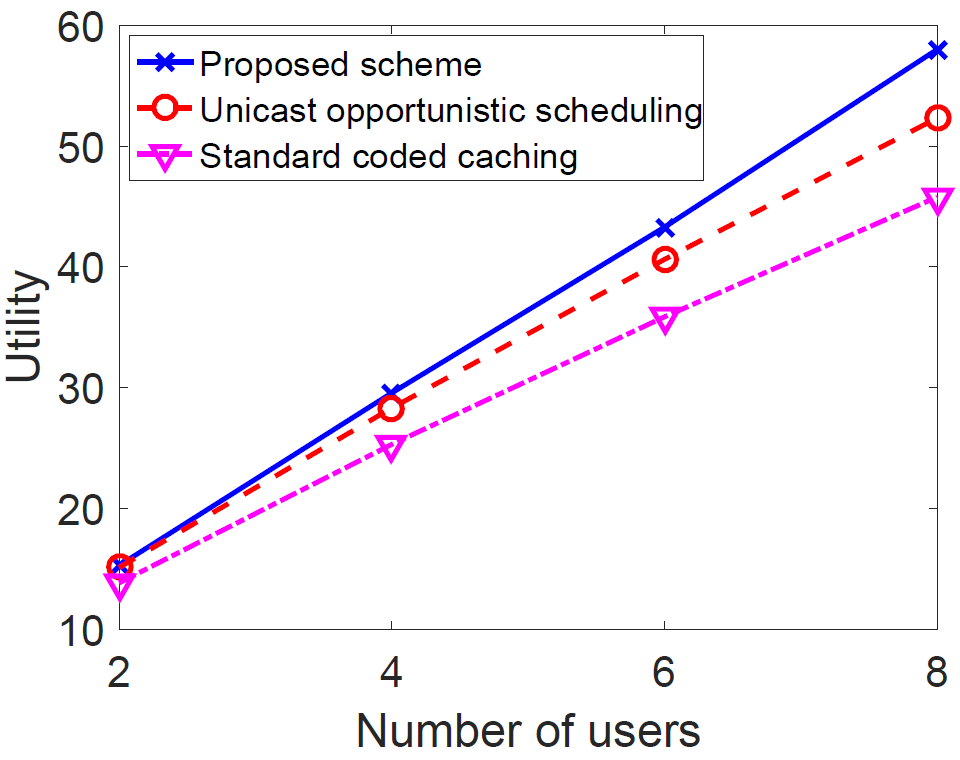}
     }
     \vspace{-2pt}
     
     \caption{Performance results vs number of users for $\alpha=0$ and $\alpha=1$}
     \label{fig:numerical_results}
     \vspace{-2pt}
   \end{figure}

Regarding the sum rate objective, standard coded caching performs very poorly, indicative of the adverse effect of users with bad channel quality. It is notable that our proposed scheme outperforms the unicast opportunistic scheme, which maximizes the sum rate if only private information packets are to be conveyed. The relative merit of our scheme increases as the number of users grows. This can be attributed to the fact that our scheme can exploit any available multicast opportunities. \emph{Our result here implies that, in realistic wireless systems, coded caching can indeed provide a significant throughput increase when an appropriate joint design of routing and opportunistic transmission is used.} 

Regarding the proportional fair objective, we can see that the average sum utility increases with a system dimension for three schemes although our proposed scheme provides a  gain compared to the two others. 


\section{Conclusions}\label{sec:conclusions}
We provided an algorithm to solve the problem of ensuring fairness in the long term delivery rates in wireless systems employing decentralized coded caching. Our results imply that appropriately combining the  opportunism arising from the fading channels with the multicasting opportunities that arise from coded caching can mitigate the harmful impact of users with bad channel conditions in standard coded caching schemes and provide significant increase in the performance of the system. 

\appendices

\section{Proof of Theorem \ref{th:capacity_BC}}\label{appendix:BC_capacity_proof}

\subsection{Converse}\label{appendix:converse}

We provide the converse proof for $K=3$ and the general case $K>3$ follows readily.
Notice that the channel output of user $k$ in \eqref{eq:BFC} for $n$ channel use can be equivalently written as 
\begin{align}
	\yv_k=\xv+\tilde{\nuv}_k,
\end{align}
where $\tilde{\nuv}_k=\frac{\nuv_k}{\sqrt{(h_k)}}\sim\Nc_{\Cc}(0, N_k\Id_n)$  for $N_k=\frac{1}{h_k}$ and $\Id_n$ identity matrix of size $n$. Since $N_1\leq N_2\leq N_3$, we set $\tilde{M}_k=\cup_{k\in\Kc\subseteq[k]} M_{\Kc}$ the message that must be decoded by user $k$ (user $k$ decodes all bits that user $k'\geq k$ decodes) at rate $\tilde{R}_k$. More explicitly, $\tilde{M}_1=\{M_1\}$, $\tilde{M}_2=\{M_2,M_{12}\}$, $\tilde{M}_3=\{M_3,M_{13},M_{23},M_{123}\}$.
By Fano's inequality, we have 
\begin{align}\label{fano}
	\begin{cases}
		nH(\tilde{M}_1)&\leq I(\tilde{M}_1;Y_1\cond \tilde{M}_2,\tilde{M}_3)\\
		nH(\tilde{M}_2)&\leq I(\tilde{M}_2;Y_2\cond \tilde{M}_3)\\
		nH(\tilde{M}_3)&\leq I(\tilde{M}_3;Y_3).
	\end{cases}
\end{align}
Consider 
\begin{align}\label{eq171}
	I(\tilde{M}_3;Y_3)=H(Y_3)-H(Y_3\cond \tilde{M}_3).
\end{align}
Since $n\log\left( 2\pi eN_3\right)=H(Y_3\cond \tilde{M}_3,X) \leq H(Y_3\cond \tilde{M}_3)\leq H(Y_3)\leq n\log\left( 2\pi e(P+N_3)\right)$, there exist $0\leq \alpha_3\leq 1$ such that 
\begin{align}\label{eq172}
	H(Y_3\cond \tilde{M}_3)=n\log\left( 2\pi e((1-\alpha_3)P+N_3)\right).
\end{align}
Using \eqref{eq171} and \eqref{eq172} we obtain
\begin{align}\label{eq173}
	\MoveEqLeft{ I(\tilde{M}_3;Y_3)}\nonumber\\
	&= H(Y_3)-H(Y_3\cond \tilde{M}_3)\nonumber\\
	&\leq n\log\left( 2\pi e(P+N_3)\right)-n\log\left( 2\pi e((1-\alpha_3)P+N_3)\right)\nonumber\\
	&=n\log\left(\frac{N_3+P}{N_3+(1-\alpha_3)P}\right).
\end{align}

Next consider 
\begin{align}\label{eq174}
	I(\tilde{M}_2;Y_2\cond \tilde{M}_3)=H(Y_2\cond \tilde{M}_3)-H(Y_2\cond \tilde{M}_2,\tilde{M}_3).
\end{align}
Using the conditional entropy power inequality in \cite{el2011network} , we have
\begin{align}\label{eq175}
	H(Y_3\cond \tilde{M}_3)&=H(Y_2+n_3-n_2\cond \tilde{M}_3)\nonumber\\
	&\geq n\log(2^{2H(Y_2\cond \tilde{M}_3)/n}+2^{2H(n_3-n_2\cond \tilde{M}_3)/n})\nonumber\\
	&= n\log(2^{2H(Y_2\cond \tilde{M}_3)/n}+2\pi e(N_3-N_2))
\end{align}
\eqref{eq172} and \eqref{eq175} imply 
\begin{align}
	\MoveEqLeft{n\log\left( 2\pi e((1-\alpha_3)P+N_3)\right)}\nonumber\\
	&\geq n\log(2^{2H(Y_2\cond \tilde{M}_3)/n}+2\pi e(N_3-N_2))\nonumber
\end{align}
equivalent to
\begin{align}\label{eq176}
	H(Y_2\cond \tilde{M}_3)&\leq n\log(2\pi e((1-\alpha_3)P+N_2)).
\end{align}
Since $n\log( 2\pi eN_2)=H(Y_2\cond \tilde{M}_2,\tilde{M}_3,X)\leq H(Y_2\cond \tilde{M}_2,\tilde{M}_3)\leq H(Y_2\cond \tilde{M}_3)$, there exists $\alpha_2$ such that $0\leq1-\alpha_2-\alpha_3\leq 1-\alpha_3$ and 
\begin{align}\label{eq177}
	H(Y_2\cond \tilde{M}_2,\tilde{M}_3)=n\log(2\pi e((1-\alpha_2-\alpha_3)P+N_2)).
\end{align}
Using \eqref{eq174}, \eqref{eq176} and \eqref{eq177} it follows
\begin{align}\label{eq178}
	I(\tilde{M}_2,\tilde{M}_3;Y_2)&=H(Y_2\cond \tilde{M}_3)-H(Y_2\cond \tilde{M}_2,\tilde{M}_3)\nonumber\\
	&\leq n\log(2\pi e((1-\alpha_3)P+N_2))\\
	\MoveEqLeft{-n\log(2\pi e((1-\alpha_2-\alpha_3)P+N_2))}\nonumber\\
	&=n\log\left(\frac{N_2+(1-\alpha_3)P}{(1-\alpha_2-\alpha_3)P+N_2}\right).
\end{align}

Last we consider
\begin{align}\label{eq179}
	\MoveEqLeft{I(\tilde{M}_1;Y_1\cond \tilde{M}_2,\tilde{M}_3)}\nonumber\\
	&=H(Y_1\cond \tilde{M}_2,\tilde{M}_3)-H(Y_1\cond \tilde{M}_1, \tilde{M}_2,\tilde{M}_3)\nonumber\\
	&=H(Y_1\cond \tilde{M}_2,\tilde{M}_3)-H(Y_1\cond \tilde{M}_1, \tilde{M}_2,\tilde{M}_3, X)\nonumber\\
	&=H(Y_1\cond \tilde{M}_2,\tilde{M}_3)-H(Y_1\cond  X)\nonumber\\
	&=H(Y_1\cond \tilde{M}_2,\tilde{M}_3)-n\log\left( 2\pi eN_1\right)
\end{align}
Using the conditional entropy power inequality in \cite{el2011network} , we have
\begin{align}\label{eq180}
	\MoveEqLeft{H(Y_2\cond \tilde{M}_2, \tilde{M}_3)}\nonumber\\
	&=H(Y_1+n_2-n_1\cond \tilde{M}_2,\tilde{M}_3)\nonumber\\
	&\geq n\log(2^{2H(Y_1\cond \tilde{M}_2,\tilde{M}_3)/n}+2^{2H(n_2-n_1\cond \tilde{M}_2,\tilde{M}_3)/n})\nonumber\\
	&= n\log(2^{2H(Y_1\cond \tilde{M}_2,\tilde{M}_3)/n}+2\pi e(N_2-N_1))
\end{align}

\eqref{eq177} and \eqref{eq180} imply 
\begin{align}
	\MoveEqLeft{n\log(2\pi e((1-\alpha_2-\alpha_3)P+N_2))}\nonumber\\
	&\geq n\log(2^{2H(Y_1\cond \tilde{M}_2,\tilde{M}_3)/n}+2\pi e(N_2-N_1))\nonumber
\end{align}
equivalent to
\begin{align}\label{eq181}
	H(Y_1\cond \tilde{M}_2,\tilde{M}_3)&\leq n\log(2\pi e((1-\alpha_2-\alpha_3)P+N_1)).
\end{align}
Let $\alpha_1=1-\alpha_2-\alpha_3$. Combining the last inequality with \eqref{eq179} we obtain 
\begin{align}\label{eq182}
	I(\tilde{M}_1;Y_1\cond \tilde{M}_2,\tilde{M}_3)\leq n\log\left( \frac{N_1+\alpha_1P}{N_1}\right). 
\end{align}
From \eqref{fano}, \eqref{eq173}, \eqref{eq178} and \eqref{eq182}, it readily follows that $\exists$ $0\leq\alpha_1,\alpha_2,\alpha_3\leq1$ such that $\alpha_1+\alpha_2+\alpha_3=1$ and 

\begin{align}
	\begin{cases}
		H(\tilde{M}_1)&\leq \log\left(1+ \frac{\alpha_1P}{N_1}\right),\nonumber\\
		H(\tilde{M}_2)&\leq \log\left(1+\frac{\alpha_2P}{N_2+\alpha_1P}\right),\nonumber\\
		H(\tilde{M}_3)&\leq \log\left(1+\frac{\alpha_3P}{N_3+(\alpha_1+\alpha_2)P}\right).
	\end{cases}
\end{align}
By replacing $H(\tilde{M}_k)$ with $\sum_{k\in\Kc\subseteq[k]}R_{\Kc}$ and $N_k$ with $\frac{1}{h_k}$ we obtain the result
\begin{align}
	\begin{cases}
		R_1&\leq \log\left( 1+h_1\alpha_1P\right)\nonumber\\
		R_2+R_{12}&\leq \log\left(\frac{1+h_2(\alpha_1+\alpha_2)P}{1+h_2\alpha_1P}\right)\nonumber\\
		R_3+R_{13}+R_{23}+R_{123}&\leq \log\left(\frac{1+h_3P}{1+h_3(\alpha_1+\alpha_2)P}\right),
	\end{cases}
\end{align}

\subsection{Achievability} \label{appendix:achiev}

Superposition coding achieves the upper bound. For $1\leq k\leq3$, generate random sequences $u^{n}_{k}(m_k)$, $m_k\in[1:2^{n\tilde{R}_k}]$ each i.i.d. $\Nc_{\Cc}(0, \alpha_kP)$. To transmit a triple message $(m_1,m_2,m_3)$ the encoder set $X=u^{n}_{1}(m_1)+u^{n}_{2}(m_2)+u^{n}_{3}(m_3)$.
For decoding:
\begin{itemize}
	\item Receiver $3$ recover $m_3$ from $Y_3=u^{n}_{3}(m_3)+\left( u^{n}_{1}(m_1)+u^{n}_{2}(m_2)+n_3\right) $ by considering $u^{n}_{1}(m_1)+u^{n}_{2}(m_2)$ as noise. The probability of error tends to zero as $n\rightarrow\infty$ if $\tilde{R}_3\leq \log\left( 1+\frac{\alpha_3P}{N_3+(\alpha_1+\alpha_2)P}\right) $.
	\item Receiver $2$ uses successive cancellation. First, it decodes $m_3$ from $Y_2=u^{n}_{3}(m_3)+\left( u^{n}_{1}(m_1)+u^{n}_{2}(m_2)+n_2\right) $ by considering $u^{n}_{1}(m_1)+u^{n}_{2}(m_2)$ as noise. The probability of error tends to zero as $n\rightarrow\infty$ if $\tilde{R}_3\leq \log\left(1+ \frac{\alpha_3P}{N_2+(\alpha_1+\alpha_2)P}\right) $. Since $N_2\leq N_3$ and $\tilde{R}_3\leq \log\left(1+\frac{\alpha_3P}{N_3+(\alpha_1+\alpha_2)P}\right)$, the later condition is satisfied. Second, it subtracts off $u^{n}_{3}(m_3)$ and recover $u^{n}_{2}(m_2)$ from $\tilde{Y}2=u^{n}_{2}(m_2)+\left( u^{n}_{1}(m_1)+n_2\right)$ by treating  $u^{n}_{1}(m_1)$ as noise. The probability of error tends to zero as $n\rightarrow\infty$ if $\tilde{R}_2\leq \log\left(1+\frac{\alpha_2P}{N_2+\alpha_1P}\right) $.
	\item Receiver $1$ uses successive cancellation twice. First, it decodes $m_3$ from $Y_1=u^{n}_{3}(m_3)+\left( u^{n}_{1}(m_1)+u^{n}_{2}(m_2)+n_1\right) $ by considering $u^{n}_{1}(m_1)+u^{n}_{2}(m_2)$ as noise. The probability of error tends to zero as $n\rightarrow\infty$ if $\tilde{R}_3\leq \log\left(1+\frac{\alpha_3P}{N_1+(\alpha_1+\alpha_2)P}\right) $. Since $N_1\leq N_3$ and $\tilde{R}_3\leq \log\left(1+\frac{\alpha_3P}{N_3+(\alpha_1+\alpha_2)P}\right)$, the later condition is satisfied. Second, it subtracts off $u^{n}_{3}(m_3)$ and decodes $u^{n}_{2}(m_2)$ by treating  $u^{n}_{1}(m_1)$ as noise. The probability of error tends to zero as $n\rightarrow\infty$ if $\tilde{R}_2\leq \log\left(1+\frac{\alpha_2P}{N_1+\alpha_1P}\right) $. Since $N_1\leq N_2$ and $\tilde{R}_2\leq\log\left(1+\frac{\alpha_2P}{N_2+\alpha_1P}\right)$, the later condition is satisfied. Last, it subtracts off $u^{n}_{2}(m_2)$ and recover  $u^{n}_{1}(m_1)$. The probability of error tends to zero as $n\rightarrow\infty$ if $\tilde{R}_1\leq \log\left(1+\frac{\alpha_1P}{N_1}\right)$. 
\end{itemize}

\section{Proof of Lemma \ref{lem:equivalence}}\label{appendix:equivalence_proof}

Denote $A_k(t)$ the number of files that have been admitted to the system for user $k$ up to slot $t$. Also, note that due to our restriction on the class of policies $\Pi^{CC}$ and our assumption about long enough blocklengths, there are no errors in decoding the files, therefore the number of files correctly decoded for user $k$ till slot $t$ is $D_k(t)$. Since $D_k(t)\leq A_k(t), \forall t\geq 0, \forall k=1,..,K$, if suffices to show that for every arrival rate vector $\bar{\av} \in \Lambda^{CC}$, there exists a policy in $\Pi^{CC}$ for which the delivery rate vector is $\bar{\rv} = \bar{\av}$. 

We shall deal only with the interior of $\Lambda^{CC}$ (arrival rates at the boundaries of stability region are exceptional cases). Take any arrival rate vector $\bar{\av}\in Int(\Lambda^{CC})$. From \cite[Theorem 4.5]{neely10} it follows that for any there exists a randomized demand combination and transmission policy $\pi^{RAND}$, the probabilities of which depending only on the channel state realization each slot, for which the system is strongly stable. In addition, any arrival rate vector can be constructed via a randomized admission policy. Since the channels are i.i.d. random with a finite state space and queues are measured in files and bits, the system now evolves as a discrete time Markov chain $(\mathbf{S}(t), \mathbf{Q}(t), \mathbf{H}(t))$, which can be checked that is aperiodic, irreducible ad with a single communicating class. In that case, strong stability means that the Markov chain is ergodic with finite mean. 

Further, this means that the system reaches to the set of states where all queues are zero infinitely often. Let $T[n]$ be the number of timeslots between the $n-$th and $(n+1)-$th visit to this set (we make the convention that $T[0]$ is the time slot that this state is reached for the first time). In addition, let $\hat{A}_k[n], \hat{D}_k[n]$ be the number of demands that arrived and were delivered in this frame, respectively. Then, since within this frame the queues start and end empty, we have \[
[\hat{A}_k[n] = \hat{D}_k[n], \forall n, \forall k.\]

In addition since the Markov chain is ergodic,  
\[
\bar{a}_k = \lim\limits_{t\rightarrow\infty}\frac{A(t)}{t} = \lim\limits_{N\rightarrow\infty}\frac{\sum_{n=0}^N\hat{A}_k[n]}{\sum_{n=0}^NT[n]}
\]
and 
\[
\bar{r}_k = \lim\limits_{t\rightarrow\infty}\frac{D(t)}{t} = \lim\limits_{N\rightarrow\infty}\frac{\sum_{n=0}^N\hat{D}_k[n]}{\sum_{n=0}^NT[n]}
\]

Combining the three expressions, $\bar{\rv} = \bar{\av}$ thus the result follows. 

\section{Proof of Theorem \ref{th:optimality_infinite}}\label{appendix:performance_proof}

From Lemma \ref{lem:equivalence} and Corollary \ref{cor:equivalentOptimization}, it suffices to prove that under the online policy the queues are strongly stable and the resulting time average admission rates maximize the desired utility function subject to minimum rate constraints. 

We first look at policies that take random decisions based only on the channel realizations. Since the  feasibility region $\Lambda^{CC}$ is a convex set (see Theorem ), any point in it can be achieved by properly time-sharing over the possible control decisions. We focus on two such policies, one that achieves the optimal utility and another on that achieves (i.e. admits and stabilizes the system for that) a rate vector in th $\delta-$ interior of $\Lambda^{CC}$. We then have the following Lemmas:  

\begin{lemma}[Static Optimal Policy]\label{lem:StaticOptimalPolicy} Define a policy $\pi^*\in\Pi^{CC}$ that in each slot where the channel states are $\mathbf{h}$ works as follows:  (i) it pulls random user demands with mean $\bar{a}_k^*$, and it gives the virtual queues arrivals with mean $\bar{\gamma}_k = \bar{a}_k^*$ as well (ii) the number of combinations for subset $\Jc$ is a random variable with mean $\bar{\sigma}^*_{\Jc}$ and uniformly bounded by $\sigma_{max}$, (iii) selects one out of $K+1$ suitably defined rate vectors $\mathbf{\mu^l}\in \Gamma(\hv), l=1,..,K+1$ with probability $\psi_{l, \hv}$. The parameters above are selected such that they solve the following problem: 
	\begin{align}\nonumber
	\max_{\bar{\av}}&\sum_{k=1}^Kg_k(\bar{a}_k^*)\\ \nonumber
	\text{s.t.}\quad & \sum_{\Jc: k\in \Jc}\bar{\sigma}_{\Jc} \geq \bar{a}_k^*, \forall k\in\{1,..,K\}\\ \nonumber
	& \sum_{\Jc: \Ic \subseteq \Jc}b_{\Jc, \Ic}\bar{\sigma}_{\Jc} \geq \sum_{\hv}\pi_{\hv}\sum_{l=1}^{K+1}\psi_{l, \hv}\mu^l_{\Ic}(\hv), \forall \Ic \in 2^{\Kc}
	\end{align}
	Then, $\pi^*$ results in the optimal delivery rate vector (when all possible policies are restricted to set $\Pi$). 	
\end{lemma}

\begin{lemma}[Static Policy for the $\delta-$ interior of $\Lambda^{CC}$]\label{lem:StablePolicy} Define a policy $\pi^{\delta}\in\Pi^{CC}$ that in each slot where the channel states are $\mathbf{h}$ works as follows:  (i) it pulls random user demands with mean $\bar{a}_k^{\delta}$ such that $(\bar{\av}+\mathbf{\delta})\in\Pi^{CC}$, and gives the virtual queues random arrivals with mean $\bar{\gamma}_k \leq \bar{\av} + 
	\epsilon'$ for some $\epsilon'>0$ (ii) the number of combinations for subset $\Jc$ is a random variable with mean $\bar{\sigma}^{\delta}_{\Jc}$ and uniformly bounded by $\sigma_{max}$, (iii) selects one out of $K+1$ suitably defined rate vectors $\mathbf{\mu^l}\in \Gamma(\hv), l=1,..,K+1$ with probability $\psi^{\delta}_{l, \hv}$. The parameters above are selected such that: 
	\begin{align}\nonumber
	& \sum_{\Jc: k\in \Jc}\bar{\sigma^{\delta}}_{\Jc} \geq \epsilon + \bar{a}_k^{\delta}, \forall k\in\{1,..,K\}\\ \nonumber
	& \sum_{\Jc: \Ic \subseteq \Jc}b_{\Jc, \Ic}\bar{\sigma}^{\delta}_{\Jc} \geq \epsilon + \sum_{\hv}\pi_{\hv}\sum_{l=1}^{K+1}\psi'_{l, \hv}\mu^l_{\Ic}(\hv), \forall \Ic \in 2^{\Kc}
	\end{align}
	for some appropriate $\epsilon < \delta$. Then, the system under $\pi^{\delta}$ has mean incoming rates of $\bar{\av}^{\delta}$ and is strongly stable. 
\end{lemma}  

The proof of the performance of our proposed policy is based on applying Lyapunov optimization theory \cite{neely10} with the following as Lyapunov function
\[
L(\Zm)= L(\mathbf{S}, \mathbf{Q}, \mathbf{U}) = \frac{1}{2}\left(\sum_{k=1}^KU_k^2(t) + S_k^2(t) + \sum_{\Ic \in 2^{\Kc}}\frac{Q_{\Ic}^2(t)}{F^2}\right)
.\]
Defining its drift as 
\[
\Delta L(\mathbf{Z}) = \mathbb{E}\left\{L(\mathbf{Z}(t+1))-L(\mathbf{Z}(t))|\mathbf{Z}(t)=\mathbf{Z}\right\}
\], using the queue evolution equations and the fact that $([x]^+)^2\leq x^2$, we have 
\begin{align}\nonumber 
\Delta L(\mathbf{Z}(t)) \leq B &+ \sum_{\Ic \in 2^{\Kc}}\frac{Q_{\Ic}(t)}{F^2}\mathbb{E}\left\{\sum_{\Jc: \Jc\supseteq \Ic}b_{\Ic,\Jc}\sigma_{\Jc}(t) - \mu_{\Ic}(t)\right\}\\ \nonumber
& + \sum_{k=1}^KS_k(t)\mathbb{E}\left\{a_k(t) - \sum_{\Ic: k\in\Ic}\sigma_{\Ic}(t)\right\}\\ \label{eq:drift1_infiniteDemand}
& + \sum_{k=1^K}U_k(t)\mathbb{E}\left\{\gamma_k(t) - a_k(t)\right\}  
,\end{align}
where $B<\infty$ is a constant that depends only on the parameters of the system. Adding the quantity $-V\sum_{k=1}^K\mathbb{E}\left\{g_k(\gamma_k(t))\right\}$ to both hands of \eqref{eq:drift1_infiniteDemand} and rearranging the right hand side, we have 
\begin{align}\nonumber
&\Delta L(\mathbf{Z}(t)) - V\sum_{k=1}^K\mathbb{E}\left\{g_k(\gamma_k(t))\right\}\leq \\ B\nonumber  &+\sum_{k=1}^K\mathbb{E}\left\{-Vg_k(\gamma_k(t)) + \gamma_k(t)U_k(t))\right\}  \\ \nonumber
&+\sum_{\Ic}\mathbb{E}\left\{\sigma_{\Ic}(t)\right\} \left(\sum_{\Jc:\Jc\subseteq \Ic}\frac{Q_{\Jc}(t)}{F^2}b_{\Ic, \Jc}- \sum_{k\in\Jc}S_k(t)\right) \\ \nonumber 
& + \sum_{k=1}^K\left(S_k(t) -  U_k(t)\right)\mathbb{E}\left\{a_k(t)\right\}  \\ \label{eq:drift2_inifiniteDemand}
&-\sum_{\Jc}\frac{Q_{\Jc}(t)}{F^2}\mathbb{E}\left\{\mu_{\Jc}(t)\right\}
\end{align} 

Now observe that the control algorithm minimizes right hand side of \eqref{eq:drift2_inifiniteDemand} given the channel state $\hv(t)$ (for any channel state). Therefore, taking expectations over the channel state distributions, for every vectors $\bar{\av}\in [1,\gamma_{max}]^K, \bar{\mathbf{\gamma}}\in [1,\gamma_{max}]^K, \bar{\mathbf{\sigma}}\in Conv(\{0,..,\sigma_{max}\}^M),  \bar{\mu}\in \sum_{\hv\in\mathcal{H}}\pi_{\hv}\Gamma(\hv)$ it holds that 
\begin{align}\nonumber
&\Delta L(\mathbf{Z}(t)) - V\sum_{k=1}^K\mathbb{E}\left\{g_k(\gamma^{\pi}_k(t))\right\}\leq \\ \nonumber 
B&- V\sum_{k=1}^Kg_k(\bar{\gamma}_k)  + \sum_{k=1}^KU_k(t)\left(\bar{\gamma}_k - \bar{a}_k\right)\\ \nonumber 
&+\sum_{k=1}^KS_k(t)\left(\bar{a}_k -\sum_{\Jc :k\in\Jc}\bar{\sigma}_{\Jc}\right)\\ \label{eq:infinite_drift3} 
& + \sum_{\Jc}\frac{Q_{\Jc}(t)}{F^2}\left(\sum_{\Ic :\Jc\subseteq \Ic}b_{\Jc,\Ic}\bar{\sigma}_{\Ic} - \bar{\mu}_{\Jc}\right)
\end{align}
We will use \eqref{eq:infinite_drift3} to compare our policy with the static policies defined in Lemmas \ref{lem:StaticOptimalPolicy}, \ref{lem:StablePolicy}.  More specifically, replacing the time averages we get from the static stabilizing policy $\pi^{\delta}$ of Lemma \ref{lem:StablePolicy} for some $\delta>0$, we get that thre exist $\epsilon >0$ such that 
\begin{align}\nonumber 
\Delta L(\mathbf{Z}(t)) \leq B&  + V\sum_{k=1}^K\mathbb{E}\left\{g_k(a^{\pi}_k(t))\right\}- V\sum_{k=1}^Kg_k(\bar{a}_k^{\delta}) \\ \nonumber 
&-\epsilon \left(\sum_{k=1}^K S_k(t) + \sum_m\frac{Q_{\Jc}(t)}{F^2}\right)\\ \label{eq:infinite_drift_4}
&\epsilon'\sum_{k=1}^KU_k(t)
\end{align}
Since $a_k(t)\leq \gamma_{max} \forall t$, it follows that $g_k(\bar{a}_k^{\delta})<g_k(\gamma_{max})$, therefore, we have from the Foster-Lyapunov criterion that the system $(\mathbf{S}(t), \mathbf{Q}(t), \mathbf{U}(t))$ has a unique stationary probability distribution, under which the mean queue lengths are finite \footnote{For the utility-related virtual queues, note that if $g_k'(0)<\infty$, then $Y_k(t)<Vg'_k(0)+\gamma_{k,max}$, i.e. their length is deterministically bounded}. Therefore the queues are strongly stable under  our proposed policy. 

We now proceed to proving the utility-delay tradeoff. 

\textbf{Proof of near optimal utility:} Here we compare $\pi$ with the static optimal policy $\pi^*$ from Lemma \ref{lem:StaticOptimalPolicy}. Since $\pi^*$ takes decisions irrespectively of the queue lengths, we can replace quantities $\bar{\mathbf{a}}, \bar{\mathbf{\sigma}},  \bar{\mu}$ with the time averages corresponding to $\pi^*$, i.e. $\bar{\mathbf{a}}^*, \bar{\mathbf{\sigma}}^*,  \bar{\mu}^*$. We thus have:  
\[
V\sum_{k=1}^K\mathbb{E}\left\{g_k(\gamma^{\pi}_k(t))\right\} \geq V \sum_{k=1}^Kg_k(\bar{a}^{*}_k) - B + \mathbb{E}\left\{\Delta L (\mathbf{Z}(t))\right\}
\]
Taking expectations over $\mathbf{Z}(t)$ for both sides and summing the inequalities for $t=0,1,..,T-1$ we get 
\begin{align} \nonumber
\frac{1}{T}\sum_{t=1}^{T-1}\sum_{k=1}^K\mathbb{E}\left\{g_k(\gamma^{\pi}_k(t))\right\} \geq \sum_{k=1}^K&g_k(\bar{a}^{*}_k) - \frac{B}{V} - \frac{\mathbb{E}\left\{L(\mathbf{Z}(0))\right\}}{VT} \\ \nonumber
&+ \frac{\mathbb{E}\left\{L(\mathbf{Z}(T))\right\}}{VT}
\end{align}
Assuming $\mathbb{E}\left\{L(\mathbf{Z}(0))\right\} < \infty$ (this assumption is standard in this line of work, for example it holds if the system starts empty), taking the limit as $T$ goes to infinity gives
\[
\lim\limits_{T\rightarrow\infty}\frac{1}{T}\sum_{t=1}^{T-1}\sum_{k=1}^K\mathbb{E}\left\{g_k(\gamma^{\pi}_k(t))\right\} \geq \sum_{k=1}^Kg_k(\bar{a}^{*}_k) - \frac{B}{V}
\]
In addition, since $g_k(x)$ are concave, Jensen's inequality implies 
\begin{align}\nonumber
\sum_{k=1}^Kg_k(\bar{\gamma}_l^{\pi})&=\sum_{k=1}^Kg_k\left(\lim_{T\rightarrow \infty}\frac{1}{T}\sum_{t=0}^T\mathbb{E}\{\gamma_k^{\pi}(t)\}\right)\\ \nonumber 
&\geq \lim\limits_{T\rightarrow\infty}\frac{1}{T}\sum_{t=1}^{T-1}\sum_{k=1}^K\mathbb{E}\left\{g_k(\gamma^{\pi}_k(t))\right\} \\ \nonumber 
&\geq \sum_{k=1}^Kg_k(\bar{a}^{*}_k) - \frac{B}{V}
.\end{align}
Proving the near optimality of the online policy follows from the above and the fact that $\bar{a}_k^{\pi} > \bar{\gamma}_k^{\pi}$ (since the virtual queues $U_k(t)$ are strongly stable).



\begin{thebibliography}{10}
\providecommand{\url}[1]{#1}
\csname url@samestyle\endcsname
\providecommand{\newblock}{\relax}
\providecommand{\bibinfo}[2]{#2}
\providecommand{\BIBentrySTDinterwordspacing}{\spaceskip=0pt\relax}
\providecommand{\BIBentryALTinterwordstretchfactor}{4}
\providecommand{\BIBentryALTinterwordspacing}{\spaceskip=\fontdimen2\font plus
\BIBentryALTinterwordstretchfactor\fontdimen3\font minus
  \fontdimen4\font\relax}
\providecommand{\BIBforeignlanguage}[2]{{%
\expandafter\ifx\csname l@#1\endcsname\relax
\typeout{** WARNING: IEEEtran.bst: No hyphenation pattern has been}%
\typeout{** loaded for the language `#1'. Using the pattern for}%
\typeout{** the default language instead.}%
\else
\language=\csname l@#1\endcsname
\fi
#2}}
\providecommand{\BIBdecl}{\relax}
\BIBdecl

\bibitem{cisco15} ``White paper: Cisco VNI Forecast and Methodology, 2015-2020'', Tech. Report, 2015.

\bibitem{maddah2013fundamental}
M. Maddah-Ali and U. Niesen, ``Fundamental Limits of Caching,'' \emph{{IEEE} Trans. Inf. Theory}, vol. 60, no. 5, pp. 2856--2867, 2014.

\bibitem{misconceptions} G.~S.~Paschos, E.~Bastug, I. Land, G. Caire, and M. Debbah, ``Wireless caching: technical misconceptions and business barriers'', IEEE Communications Magazine, 2016.
  
\bibitem{ji2013fundamental}
M. Ji,  G. Caire, A. Molisch,  ``Fundamental Limits of Distributed Caching in D2D Wireless Networks'' , arXiv/1304.5856, 2013.

\bibitem{ji2015order}
M. Ji,  A. Tulino, J. Llorca, and G. Caire,  ``Order-Optimal Rate of Caching and Coded Multicasting with Random Demands'', arXiv:1502.03124, 2015. 

\bibitem{maddah2013decentralized}
M. Maddah-Ali and U. Niesen, ``Decentralized Coded Caching Attains Order-Optimal Memory-Rate Tradeoff'', 
 \emph{{IEEE/ACM} Trans. Netw.}, vol. 23, no. 4, pp. 1029--1040, 2015.



\bibitem{niesen2013coded}
M. Maddah-Ali and U. Niesen, ``Coded Caching with Nonuniform Demands'',\newblock in {\em IEEE INFOCOM Workshops}, 2014.

\bibitem{pedarsani2016online}
R. Pedarsani, M. Maddah-Ali, and  U. Niesen, ``Online Coded Caching," in IEEE/ACM Trans. Netw., vol. 24, no. 2, pp. 836-845, 2016.

 


\bibitem{zhang2016wireless}
J. Zhang, and P. Elia, "Wireless Coded Caching: a Topological Perspective". arXiv:1606.08253, 2016.

\bibitem{bidokhti2016noisy}
 S. S. Bidokhti, M. Wigger, and R. Timo, ``Noisy Broadcast Networks with Receiver Caching'', arXiv preprint arXiv:1605.02317, 2016. 
 
\bibitem{NgoAllerton2016}
 K-H. Ngo, S. Yang, and M. Kobayashi, ``{Cache-Aided Content Delivery in MIMO Channels}'', 
 \newblock in {\em Proc. Allerton}, IL, USA, 2016.


\bibitem{niesen2015coded} U. Niesen M. Maddah-Ali ``Coded Caching for Delay-Sensitive Content'', in IEEE ICC, pp. 5559-5564, 2015.


\bibitem{el2011network}
A. El Gamal and Y. H. Kim, ``{Network Information Theory}'', Cambridge university press, 2011.

  \bibitem{neely10} M. Neely, ``Stochastic Network Optimization with Application to Communication and Queueing Systems'', Morgan \& Claypool,  2010.
  
    \bibitem{mowalrand} J.~Mo and J.~Walrand, ``Fair end-to-end window-based congestion control'', IEEE/ACM Trans. Netw. , Vol. 8, No. 5, Oct. 2000. 
    
    \bibitem{pfscheduling}    F. Kelly, ``Charging and Rate Control for Elastic Traffic'', European Transactions on Telecommunications, 1997.
    
     \bibitem{stolyar} A.~L.~Stolyar,  ``On the asymptotic optimality of the gradient scheduling algorithm for multiuser throughput allocation." Operations Research Vol. 53 No. 1, 2005.

      \bibitem{knopp} R. Knopp and P. A. Humblet, ``Information capacity and power control in single-cell multiuser communications," in IEEE ICC, Seattle, WA, 1995.

       \bibitem{caire_fairnessMIMO} H. Shirani-Mehr, G. Caire and M. J. Neely, ``MIMO Downlink Scheduling with Non-Perfect Channel State Knowledge," in IEEE Trans. Commun. , vol. 58, no. 7, pp. 2055-2066, July 2010.

       \bibitem{caire_fairnessBC} G. Caire, R. R. Muller and R. Knopp, ``Hard Fairness Versus Proportional Fairness in Wireless Communications: The Single-Cell Case," in IEEE Trans. Inf. Theory, vol. 53, no. 4, pp. 1366-1385, April 2007.
       
       \bibitem{Seong06} K. Seong, R. Narasimhan, and J. Cioffi, ``Queue Proportional Scheduling via Geometric Programming in Fading Broadcast Channels'', in IEEE JSAC, vol. 24, no. 8, Aug 2006.
       
\bibitem{Eryilmaz01}  A. Eryilmaz, and R. Srikant, and J. R. Perkins, ``Throughput-optimal Scheduling for Broadcast Channels'', in Proc.  ITCom, Denver, CO, August 2001.

\bibitem{TseOptimal} 
D. Tse,  ``{Optimal Power Allocation over Parallel Gaussian Broadcast Channels}'', \emph{unpublished}


\bibitem{Li05} M. Neely, E. Modiano, and C.-P. Li, ``Fairness and Optimal Stochastic Control for Heterogeneous Networks'', IEEE/ACM Trans. Netw., 2005.


  
 
 \end{thebibliography}
\end{document}